\documentclass[letterpaper, 10 pt, conference]{ieeeconf}

\usepackage{booktabs}
\usepackage[usenames,dvipsnames,svgnames,table]{xcolor}
\usepackage{amssymb,latexsym,amsfonts,amsmath,bbm}
\usepackage[cal=cm,bb=pazo]{mathalfa}
\usepackage{fancyhdr}
\usepackage{subcaption}
\usepackage{mathrsfs}
\usepackage{dsfont}
\usepackage{mathtools}
\usepackage{graphicx}
\usepackage{eso-pic}
\usepackage{epsfig}
\usepackage{mathtools}
\usepackage{tikz}
\usetikzlibrary{automata}
\usetikzlibrary{shapes}
\usetikzlibrary{calc,shapes,arrows}
\usepackage{mdwlist}
\usepackage{refcount}
\usepackage{comment}
\usepackage{ragged2e}
\usepackage[noadjust]{cite}

\usepackage[nocomma]{optidef}

\let\labelindent\IEEElabelindent
\usepackage{enumitem}

\makeatletter
\def\endthebibliography{%
	\def\@noitemerr{\@latex@warning{Empty `thebibliography' environment}}%
	\endlist
}
\makeatother

\usepackage{keyval,trig}
\usepackage{amssymb,dsfont}
\usepackage{mathtools,mathrsfs}

\definecolor{mediumblue}{rgb}{0.0, 0.0, 0.8}
\definecolor{mediumcandyapplered}{rgb}{0.89, 0.02, 0.17}
\definecolor{nazar}{rgb}{0.7, 0.5, 0.9}
\makeatletter
\let\NAT@parse\undefined
\makeatother
\usepackage[colorlinks=true, citecolor=mediumblue, linkcolor=mediumblue, urlcolor = mediumblue, final]{hyperref}

\DeclareRobustCommand\sampleline[1]{%
	\tikz\draw[#1] (0,0) (0,\the\dimexpr\fontdimen22\textfont2\relax)
	-- (2em,\the\dimexpr\fontdimen22\textfont2\relax);%
}

\DeclareRobustCommand\sampleline[1]{%
	\tikz\draw[#1] (0,0) (0,\the\dimexpr\fontdimen22\textfont2\relax)
	-- (1.5em,\the\dimexpr\fontdimen22\textfont2\relax);%
}

\DeclareFontFamily{U}{stix2bb}{}
\DeclareFontShape{U}{stix2bb}{m}{n} {<-> stix2-mathbb}{}

\NewDocumentCommand{\stixbbdigit}{m}{%
	\text{\usefont{U}{stix2bb}{m}{n}#1}%
}
\newcommand{\bbzero}{\stixbbdigit{0}}

\usepackage{algorithm}
\usepackage{algorithmic}

\usepackage{tikz}
\usetikzlibrary{calc,shapes,arrows}

\usepackage{xspace}

\newtheorem{theorem}{Theorem}[section]
\newtheorem{lemma}[theorem]{Lemma}

\newtheorem{definition}[theorem]{Definition}

\newtheorem{remark}[theorem]{Remark}

\newtheorem{assumption}{Assumption}

\usepackage[many]{tcolorbox}
\usetikzlibrary{calc}
\tcbuselibrary{skins}

\newtcolorbox{resp}[1][]{%
	enhanced jigsaw,%
	colback=gray!5!white,%
	colframe=gray!80!black,%
	size=small,%
	boxrule=1pt,%
	halign title=flush center,%
	coltitle=black,%
	breakable,%
	drop shadow=black!50!white,%
	attach boxed title to top left={xshift=1cm,yshift=-\tcboxedtitleheight/2,yshifttext=-\tcboxedtitleheight/2},%
	minipage boxed title=3cm,%
	boxed title style={%
		colback=white,%
		size=fbox,%
		boxrule=1pt,%
		boxsep=2pt,%
		underlay={%
			\coordinate (dotA) at ($(interior.west) + (-0.5pt,0)$);
			\coordinate (dotB) at ($(interior.east) + (0.5pt,0)$);
			\begin{scope}[gray!80!black]
				\fill (dotA) circle (2pt);
				\fill (dotB) circle (2pt);
			\end{scope}
		}%
	},%
	#1%
}

\DeclareRobustCommand{\legendsquare}[1]{%
	\textcolor{#1}{\rule{1.5ex}{1.5ex}}%
}

\definecolor{START}{rgb}{0.15,0.35,0.85}
\definecolor{TARGET}{rgb}{0.10,0.60,0.25}
\definecolor{OBSTACLES}{rgb}{0.55,0.08,0.08}

\IEEEoverridecommandlockouts

\makeatletter
\def\@opargbegintheorem#1#2#3{\textit{#1\ #2} \textit{(#3):}}
\makeatother

\newcommand{\R}{{\mathbb{R}}}
\newcommand{\Rpz}{{\mathbb{R}_{\geq 0}}}
\newcommand{\Rp}{{\mathbb{R}_{> 0}}}
\newcommand{\N}{{\mathbb{N}_{\geq 0}}}
\newcommand{\Np}{{\mathbb{N}_{\geq 1}}}
\newcommand{\I}{{\mathbb{I}}}
\newcommand{\One}{{\mathbb{1}}}

\title{From Noisy Data to Hierarchical Control: A Model-Order-Reduction Framework
\thanks{H. Sandberg was supported in part by the Swedish Research Council Grant 2023-04770 and Wallenberg AI, Autonomous Systems and Software Program (WASP DYNACON) funded by the Knut and Alice Wallenberg Foundation. K. H. Johansson was supported in part by Swedish Research Council Distinguished Professor Grant 2017-01078, Knut and Alice Wallenberg Foundation Wallenberg Scholar Grant, and Swedish Strategic Research Foundation FUSS SUCCESS Grant.}
\thanks{B. Samari and A. Lavaei are with the School of Computing, Newcastle University, NE4 5TG Newcastle upon Tyne, United Kingdom (e-mails: {\tt\small{\{b.samari2,abolfazl.lavaei\}@newcastle.ac.uk}}).}
\thanks{H. Sandberg and K. H. Johansson are with the Department of Decision and Control Systems, KTH Royal Institute of Technology, SE-100 44 Stockholm, Sweden. They are also affiliated with Digital Futures (e-mails: {\tt\small{\{hsan,kallej\}@kth.se}}).}
}

\author{Behrad Samari, Henrik Sandberg, Karl H. Johansson, and Abolfazl Lavaei
}

\allowdisplaybreaks
\begin{document}
\maketitle

\begin{abstract}
	This paper develops a direct data-driven framework for constructing reduced-order models (ROMs) of discrete-time linear dynamical systems with unknown dynamics and process disturbances. The proposed scheme enables controller synthesis on the ROM and its refinement to the original system via an interface function designed using noisy data. To achieve this, the notion of simulation functions (SFs) is employed to establish a formal relation between the original system and its ROM, yielding a quantitative bound on the mismatch between their output trajectories. To construct such relations and interface functions, we rely on data collected from the unknown system. In particular, using noise-corrupted input--state data gathered along a single trajectory of the system, and without identifying the original dynamics, we propose data-dependent conditions, cast as a semidefinite program, for the simultaneous construction of ROMs, SFs, and interface functions. Through a case study, we demonstrate that data-driven controller synthesis on the ROM, combined with controller refinement via the interface function, enables the satisfaction of complex logic specifications.
\end{abstract}

\section{Introduction}\phantomsection\label{Sec:IN}
Model order reduction (MOR) provides a systematic approach for constructing a low-dimensional model that preserves the essential and dominant characteristics of a high-dimensional dynamical system. Such reduced-order models (ROMs) facilitate system analysis and controller synthesis by reducing the computational burden associated with high-dimensional models, particularly in large-scale engineering applications~\cite{antoulas2005approximation}.

Despite these advantages, constructing reliable ROMs remains demanding, especially when the system is subject to process disturbances. The problem becomes even more challenging when an explicit model of the underlying system is unavailable, a scenario that is often unavoidable in practice. In this context, the literature adopts two main perspectives~\cite{dorfler2022bridging}: \emph{(i)} combining system identification with model-based approaches to obtain ROMs, known as indirect frameworks (see, \emph{e.g.},~\cite{ASAD2025TAC}), or \emph{(ii)} designing ROMs directly from data without an intermediate modeling step, known as direct approaches. Bypassing the system identification stage can offer significant advantages, particularly when the system cannot be uniquely modeled~\cite{van2020dataTAC}, which is typically the case when the unknown system is subject to process disturbances. However, establishing correctness guarantees becomes more challenging, since all conclusions should be drawn directly from data without relying on an explicit model.

\textbf{Related Literature and Core Contribution.} MOR techniques, whether employed in model-based or data-driven settings, can be broadly categorized into three main groups, with our approach falling under the third category: \emph{(i)} energy-based approaches, such as balanced truncation~\cite{moore1981principal,prajna2005model,sandberg2010extension,besselink2014model} and Hankel-norm methods~\cite{glover1984all,kawano2016model}; \emph{(ii)} Krylov-based methods, which rely on interpolation and/or moment matching~\cite{astolfi2010model,shakib2023time,moreschini2024closed,shakib2025dissipativity}; and \emph{(iii)} approaches grounded in the notion of simulation functions (SFs), in which a similarity relation between the original high-dimensional system and its corresponding ROM is established~\cite{girard2009hierarchical,zamani2017compositional,lavaei2019compositional,lavaei2020compositional}.

Several insightful data-driven MOR results have been developed within the first two categories. In the first category,~\cite{rapisarda2011identification} proposes a balanced truncation method based on persistently exciting data,~\cite{gosea2022data} introduces a reformulation of balanced truncation via the estimation of Gramian-related quantities, and~\cite{burohman2023data} presents a framework operating on noisy data under a known noise model. In the second category,~\cite{scarciotti2017data} proposes a data-driven MOR framework based on time-domain measurements. Using swapped interconnection,~\cite{mao2022data} develops an algorithm that asymptotically estimates an arbitrary number of moments from time-domain measurements. More recently,~\cite{bhattacharjee2025signal} constructs ROMs from input--output data that achieve moment matching even when both the system and the signal generator are unknown, and~\cite{scandella2025artificial} explores artificial neural networks for data-driven moment matching.

While promising, the aforementioned studies are primarily suited for stability and input--output behavior analysis. In such frameworks, the same input is often applied to both the original system and its ROM, meaning that their formulations preserve the input dimension and generally do not allow its reduction. Moreover, these approaches typically do not provide guaranteed bounds on the mismatch between the output trajectories of the original system and its ROM at every time step, which can limit their direct applicability to enforcing complex specifications, expressed using linear temporal logic (LTL) formulas~\cite{baier2008principles}, \emph{e.g.}, safety, reachability, and reach-while-avoid.
In contrast, approaches in the third category permit different inputs for the original system and its ROM while providing trajectory-wise bounds on their output mismatch at every time step. Specifically, a control input designed for the ROM is refined through an interface function into a control input for the original system, enabling the latter to satisfy the specification fulfilled by the ROM~\cite{lavaei2022automated}.

In the third category, to the best of our knowledge, this paper proposes the first direct data-driven method for constructing ROMs of discrete-time linear dynamical systems with unknown dynamics and process disturbances, while also establishing formal relations between the original system and its ROM using SFs and designing interface functions for controller refinement.
While the work in~\cite{11312476} also belongs to the third category, it focuses on continuous-time linear dynamical systems and does not account for process disturbances or noise in the collected data. In contrast, the present framework explicitly incorporates process disturbances, which naturally give rise to noise-corrupted data. More importantly, while~\cite{11312476} combines direct and indirect data-driven techniques (indirectly identifying the system matrices required for the interface function and certain equality conditions, while directly learning the SF from data), our proposed approach is fully direct and does not involve any system identification. Finally, our framework requires only a single set of noise-corrupted input--state data, whereas~\cite{11312476} relies on two datasets, one of which should be collected under zero input, resulting in more demanding data requirements. We also note that, while~\cite{samari2025dataARXIV} falls within the third category and considers noise-corrupted data, it is tailored to continuous-time systems and does not account for process disturbances, unlike the current framework. Accordingly, the proof steps and the main constraint in~\eqref{eq:SDP_con6} are substantially distinct from those in~\cite{samari2025dataARXIV}. Moreover, the closeness guarantee in~\eqref{eq:closeness} and the associated parameters in~\eqref{eq:rho_psi} are derived differently from those in~\cite{samari2025dataARXIV}.

\textbf{Notation.}
The identity matrix of dimension $n \times n$ is denoted by $\I_n$, $\One_n$ represents the $n$-dimensional vector whose entries are all equal to one, and $\bbzero_{n \times m}$ denotes the zero matrix of dimension $n \times m$.
Given vectors $x_i \in \R^n$, $i \in \{1, \ldots, N\}$, the matrix $X=[x_1 \,\, \ldots \,\, x_N]$ has dimension $n \times N$.
The Euclidean norm of a vector $x\in\R^{n}$ is denoted by $\vert x\vert$, whereas $\Vert P \Vert$ denotes the induced $2$-norm of a matrix $P$.
For a symmetric matrix $P$, the notation $P \succ 0$ ($P \succeq 0$) signifies that $P$ is positive (semi)definite, whereas $P \prec 0$ ($P \preceq 0$) denotes that $P$ is negative (semi)definite.
For a symmetric matrix $P$, the smallest and largest eigenvalues are represented by $\lambda_{\min}(P)$ and $\lambda_{\max}(P)$, respectively.
In a symmetric block matrix, the symbol $\star$ denotes the transpose of the corresponding off-diagonal block.
The rank of a matrix $A$ is denoted by $\operatorname{rank}(A)$.
For a given matrix $A$, $\mathrm{a} \in \operatorname{col}(A)$ indicates that the vector $\mathrm{a}$ belongs to the column space of $A$.
For a function $f : \N \to \R^n$, we define $\vert f \vert_{\infty} \coloneq \sup_{k \in \N} \vert f(k) \vert$.

\section{Problem Formulation}\phantomsection\label{Sec:PF}

\subsection{System and ROM Descriptions}\phantomsection\label{Subsec:SYS_ROM}
We consider a class of discrete-time linear control systems (dt-LCSs) subject to process disturbances, formally defined as follows.

\begin{definition}\phantomsection\label{def:system}
	A dt-LCS is specified by the tuple $\Sigma = (\mathds{X}, \mathds{U}, \mathds{Y}, \mathds{W}, A, B, \I_n)$, where $\mathds{X}=\R^n$, $\mathds{U}=\R^m$, $\mathds{Y}=\R^n$, and $\mathds{W}\subseteq\R^n$ denote the sets of states, inputs, outputs, and disturbances of the system, respectively. The matrices $A \in \R^{n \times n}$ and $B \in \R^{n \times m}$ represent the \emph{unknown} system and input matrices, respectively. We assume that all state variables of $\Sigma$ are directly measurable, as reflected by the output matrix $\I_n$.
	For a given initial state $x(0) = \mathrm{x} \in \mathds{X}$, input sequence $\nu : \N \to \mathds{U}$, and disturbance sequence $\varpi : \N \to \mathds{W}$, the state and output of $\Sigma$ evolve for all $k \in \N$ according to
	\begin{align}
		\Sigma : \begin{cases}
			\begin{array}{lll}
				\hspace{-0.2cm} x (k + 1) & \hspace{-0.3cm} = & \hspace{-0.25cm} A x(k) + B \nu(k) + \varpi (k),\\
				\hspace{-0.2cm} y (k)   & \hspace{-0.3cm} = & \hspace{-0.25cm} x(k).
			\end{array}
		\end{cases}\phantomsection \label{eq:dt-LCS}
	\end{align}
	For any initial state $\mathrm{x} \in \mathds{X}$, input sequence $\nu : \N \to \mathds{U}$, and disturbance sequence $\varpi : \N \to \mathds{W}$, the sequence $x_{\mathrm{x} \nu \varpi} : \N \to \mathds{X}$ denotes the resulting state trajectory of $\Sigma$, while $y_{\mathrm{x} \nu \varpi} = x_{\mathrm{x} \nu \varpi}$ denotes the associated output trajectory.
\end{definition}

The dt-LCS $\Sigma$ in Definition~\ref{def:system} is construed as a system with an unknown model since the matrices $A$ and $B$ are unknown. The following definition formally introduces the ROM associated with the dt-LCS $\Sigma$.

\begin{definition}\phantomsection\label{def:ROM}
	A ROM of the dt-LCS $\Sigma$ in Definition~\ref{def:system} is represented by the tuple $\hat \Sigma = (\hat{ \mathds{X}}, \hat{ \mathds{U}}, \hat{ \mathds{Y}}, \hat A, \hat B, \hat C)$, where $\hat{ \mathds{X}} \subset \R^{\hat n}$, $\hat{\mathds{U}} \subset \R^{\hat m}$, and $\hat{\mathds{Y}} \subset \R^{n}$ denote the compact state, input, and output sets of the ROM, respectively. The matrices $\hat A \in \R^{\hat n \times \hat n}$, $\hat B \in \R^{\hat n \times \hat m}$, and $\hat C \in \R^{n \times \hat n}$ correspond to the system, input, and output matrices that are to be constructed, with potentially $\hat n \ll n$.
	Given an initial state $\hat x(0) = \hat{\mathrm{x}} \in \hat{\mathds{X}}$ and an input sequence $\hat \nu : \N \to \hat{\mathds{U}}$, the state and output of $\hat \Sigma$ evolve for all $k \in \N$ according to
	\begin{align}
		\hat \Sigma : \begin{cases}
			\begin{array}{lll}
				\hspace{-0.2cm} \hat x (k + 1) & \hspace{-0.3cm} = & \hspace{-0.25cm} \hat A \hat x(k) + \hat B \hat \nu(k),\\
				\hspace{-0.2cm} \hat y (k)   & \hspace{-0.3cm} = & \hspace{-0.25cm} \hat C \hat x(k).
			\end{array}
		\end{cases}\phantomsection \label{eq:ROM}
	\end{align}
	For any initial condition $\hat{\mathrm{x}} \in \hat{\mathds{X}}$ and input sequence $\hat\nu : \N \to \hat{\mathds{U}}$ admissible\footnote{An input sequence is admissible from a given initial condition if the state sequence it generates according to~\eqref{eq:ROM} remains in $\hat{\mathds{X}}$ at all times. Throughout the paper, ROM input sequences are understood to be admissible from their specified initial conditions.} from $\hat{\mathrm{x}}$, the sequences $\hat{x}_{\hat{\mathrm{x}}\hat\nu} : \N \to \hat{\mathds{X}}$ and $\hat{y}_{\hat{\mathrm{x}}\hat\nu} : \N \to \hat{\mathds{Y}}$ denote the state and output trajectories of $\hat\Sigma$.
\end{definition}

With the definitions of the dt-LCS $\Sigma$ and its ROM $\hat{\Sigma}$ in place, the following subsection establishes a relation between them using the notion of SFs, which is subsequently utilized to derive a quantitative bound on the closeness of their output trajectories.

\subsection{Simulation Functions}\phantomsection\label{Subsec:SF}
SFs are defined on the Cartesian product of $\mathds{X}$ and $\hat{\mathds{X}}$ to measure the proximity between the output trajectories of $\Sigma$ and $\hat\Sigma$. Specifically, SFs can be used to guarantee that the mismatch between the output trajectories of the two systems remains within a prescribed error bound. We formally present this notion in the subsequent definition, adapted from the one in~\cite{1369393}.

\begin{definition}\phantomsection\label{def:SF}
	Given $\Sigma = (\mathds{X}, \mathds{U}, \mathds{Y}, \mathds{W}, A, B, \I_n)$ and its ROM $\hat \Sigma = (\hat{ \mathds{X}}, \hat{ \mathds{U}}, \hat{ \mathds{Y}}, \hat A, \hat B, \hat C)$, a function $\mathcal{S} : \mathds{X} \times \hat{ \mathds{X}} \to \Rpz$ is called an SF from $\hat \Sigma$ to $\Sigma$ if there exist constants $\alpha \in \Rp$, $\rho, \psi \in \Rpz$, and $0 < \kappa < 1$, such that
	\begin{subequations}\phantomsection\label{eq:SF}
		\begin{itemize}
			\item $\forall x \in \mathds{X}, \forall \hat x \in \hat{ \mathds{X}},$
			\begin{align}
				\alpha \vert x - \hat C \hat x \vert^2 \le \mathcal{S}(x,\hat x),\phantomsection\label{eq:SF1}
			\end{align}
			\item $\forall x\in \mathds{X}, \forall\hat x\in\hat{ \mathds{X}}, \forall \hat u \in \hat{ \mathds{U}}, \exists u \in \mathds{U}$ such that, for all $w \in \mathds{W}$, one has
			\begin{align}
				\mathcal{S}(x^+, \hat x^+) \leq \kappa \mathcal{S}(x,\hat x) + \rho \,\vert \hat u \vert^2 + \psi, \phantomsection\label{eq:SF2}
			\end{align}
			with $x^+ \coloneq Ax + Bu + w$ and $\hat x^+ \coloneq \hat A \hat x + \hat B \hat u$.\footnote{Throughout the paper, $\nu$, $\hat{\nu}$, and $\varpi$ represent input and disturbance sequences, whereas $u$, $\hat{u}$, and $w$ denote the corresponding generic pointwise values. Moreover, for every $\hat x\in\hat{\mathds{X}}$, any quantification over a pointwise ROM input $\hat u\in\hat{\mathds{U}}$ is understood to include only input values satisfying $\hat A\hat x+\hat B\hat u\in\hat{\mathds{X}}$.}
		\end{itemize}
	\end{subequations}
\end{definition}

We note that condition~\eqref{eq:SF2} requires the existence of a function $g_{k} : \mathds{X} \times \hat{\mathds{X}} \times \hat{\mathds{U}} \to \mathds{U}$ such that the generated input $u = g_{k}(x,\hat{x},\hat{u})$ fulfills this condition. Such a function, referred to as the interface function, establishes a link between the input sequences $\nu$ and $\hat{\nu}$, thereby enabling the refinement of a control input designed for $\hat{\Sigma}$ into an input applicable to $\Sigma$. In Section~\ref{Sec:DD}, we design such an interface function purely from data as one of the main contributions of the paper.
We also note that the notion of SFs introduced in Definition~\ref{def:SF} essentially conveys that when the output trajectories of $\Sigma$ and $\hat{\Sigma}$ originate from initial points that are adequately close (cf., condition~\eqref{eq:SF1}), they retain closeness as time evolves (cf., condition~\eqref{eq:SF2})~\cite{tabuada2009verification}.

The following theorem highlights the significance of the notion of SFs by formally quantifying the error between the output trajectories of $\Sigma$ and $\hat{\Sigma}$.

\begin{theorem}\phantomsection\label{thm:closeness}
	Given $\Sigma = (\mathds{X}, \mathds{U}, \mathds{Y}, \mathds{W}, A, B, \I_n)$ and its ROM $\hat \Sigma = (\hat{ \mathds{X}}, \hat{ \mathds{U}}, \hat{ \mathds{Y}}, \hat A, \hat B, \hat C)$, let $\mathcal{S}$ be an SF from $\hat{\Sigma}$ to $\Sigma$, with the parameters $\alpha$, $\rho$, $\psi$, and $\kappa$ as in Definition~\ref{def:SF}. Then, for any initial states $\mathrm{x} \in \mathds{X}$ and $\hat{ \mathrm{x}} \in \hat{\mathds{X}}$, any disturbance sequence $\varpi : \N \to \mathds{W}$, and any input sequence $\hat \nu : \N \to \hat{\mathds{U}}$ admissible from $\hat{\mathrm{x}}$, there exists an input sequence $\nu : \N \to \mathds{U}$ such that, for all $k \in \N$:
	\begin{align}
		\vert y_{\mathrm{x} \nu \varpi}(k) - \hat{ y}_{\hat{ \mathrm{x}} \hat \nu}(k) \vert \leq \sqrt{ \frac{ \mathcal{S}(\mathrm{x}, \hat{ \mathrm{x}})}{\alpha} + \frac{\rho \vert \hat \nu \vert_\infty^2 + \psi}{\alpha (1 - \kappa)}}. \phantomsection\label{eq:closeness}
	\end{align}
\end{theorem}

\begin{proof}
	We first establish the existence of an input sequence $\nu : \N \to \mathds{U}$ for the original system $\Sigma$ corresponding to the given input sequence $\hat{\nu} : \N \to \hat{\mathds{U}}$ of the ROM $\hat{\Sigma}$. This follows from condition~\eqref{eq:SF2}. In particular, at each time step $k \in \N$, given the current states $x(k)$ and $\hat{x}(k)$, as well as the input $\hat{\nu}(k)$, condition~\eqref{eq:SF2} ensures the existence of an input $\nu(k)$ satisfying the required inequality. By selecting $\nu(k)$ in this way at every time step, one obtains a well-defined input sequence $\nu : \N \to \mathds{U}$ along the evolution of the trajectories. This construction remains valid for the disturbance sequence $\varpi$, since condition~\eqref{eq:SF2} holds for all disturbances in $\mathds{W}$.
    To continue the proof, let $V(k)\coloneq \mathcal{S}(x(k),\hat x(k))$ for all $k \in \N$, for the sake of notational simplicity. According to condition~\eqref{eq:SF2}, one has
	\[
	V(k+1) = \mathcal{S}(x(k+1),\hat x(k+1))
	\leq \kappa V(k) \! + \! \rho \vert\hat \nu(k)\vert^2 \! + \! \psi. 
	\]
   Since for all $k \in \N$, one has $\vert\hat \nu(k)\vert \leq \vert\hat \nu\vert_{\infty}$, we get
	\begin{align}
		V(k+1) \leq \kappa V(k) + \rho \vert\hat \nu\vert_{\infty}^2 + \psi, \qquad \forall k \in \N.
		\phantomsection\label{eq:proof_uniform_recursion}
	\end{align}
	By iterating~\eqref{eq:proof_uniform_recursion}, for every $k \geq 1$, one gets
	\begin{align*}
		V(k)
		& \leq \kappa V(k-1) + \rho \vert\hat \nu\vert_{\infty}^2 + \psi\\
		& \leq \kappa \bigl(\kappa V(k-2) + \rho \vert\hat \nu\vert_{\infty}^2 + \psi\bigr) + \rho \vert\hat \nu\vert_{\infty}^2 + \psi\\
		& \ \vdots\\
		& \leq \kappa^k V(0) + \sum_{j=0}^{k-1} \kappa^{k-1-j}\bigl(\rho \vert\hat \nu\vert_{\infty}^2 + \psi\bigr)\\
		& = \kappa^k \mathcal{S}(\mathrm{x},\hat{\mathrm{x}})
		+ \bigl(\rho \vert\hat \nu\vert_{\infty}^2 + \psi\bigr)\sum_{j=0}^{k-1}\kappa^j.
	\end{align*}
	Since $0 < \kappa < 1$, one has
	$
	\sum_{j=0}^{k-1}\kappa^j = \frac{1-\kappa^k}{1-\kappa} \leq \frac{1}{1-\kappa}
	$, and $\kappa^k \leq 1$.
	Therefore, we get
	\begin{align}
		V(k) \leq  \mathcal{S}(\mathrm{x},\hat{\mathrm{x}}) + \frac{\rho \vert\hat \nu\vert_{\infty}^2 + \psi}{1-\kappa},
		\qquad \forall k \in \N.
		\phantomsection\label{eq:proof_bound_on_V}
	\end{align}
	On the other hand, by~\eqref{eq:SF1}, for all $k \in \N$, we have
	\begin{align}
		\alpha \vert x(k)-\hat C\hat x(k)\vert^2 \leq \mathcal{S}(x(k),\hat x(k)) = V(k).
		\phantomsection\label{eq:proof_output_bound}
	\end{align}
	Using the output equations in~\eqref{eq:dt-LCS} and~\eqref{eq:ROM}, namely, $y(k)=x(k)$ and $\hat y(k)=\hat C\hat x(k)$, inequality~\eqref{eq:proof_output_bound} becomes
	\[
	\alpha \vert y(k)-\hat y(k)\vert^2 \leq V(k).
	\]
	Combining this with~\eqref{eq:proof_bound_on_V} yields
	\[
	\alpha \vert y(k)-\hat y(k)\vert^2
	\leq \mathcal{S}(\mathrm{x},\hat{\mathrm{x}})
	+\frac{\rho \vert\hat \nu\vert_{\infty}^2 + \psi}{1-\kappa},
	\]
	and 
	therefore, one has
	\[
	\vert y(k)-\hat y(k)\vert
	\leq
	\sqrt{
		\frac{\mathcal{S}(\mathrm{x},\hat{\mathrm{x}})}{\alpha}
		+\frac{\rho \vert\hat \nu\vert_{\infty}^2 + \psi}{\alpha(1-\kappa)}
	},
	\qquad \forall k \in \N.
	\]
	Recall that $y(k)=y_{\mathrm{x}\nu\varpi}(k)$ and $\hat y(k)=\hat y_{\hat{\mathrm{x}}\hat \nu}(k)$, thereby concluding the proof.
\end{proof}

\begin{remark}\phantomsection\label{rem:difference}
The term $\psi$ in~\eqref{eq:SF2} allows one to explicitly account for the effect of noise in the collected data caused by the process disturbances acting on the dt-LCS $\Sigma$. However, large values of $\psi$ directly increase the certified upper bound in~\eqref{eq:closeness} on the mismatch between the output trajectories of $\Sigma$ and $\hat{\Sigma}$. Hence, reducing $\psi$ is an important design objective, as it potentially leads to a tighter closeness guarantee. Similarly, smaller values of $\rho$ and larger values of $\alpha$ may further tighten the bound in~\eqref{eq:closeness}. These objectives are considered in Theorem~\ref{thm:main}. In addition, enforcing $\mathcal{S}(\mathrm{x},\hat{\mathrm{x}})=0$ removes the first term in~\eqref{eq:closeness}, yielding a tighter closeness guarantee.
\end{remark}

While an SF yields a formal bound on the closeness between the output trajectories of $\Sigma$ and $\hat{\Sigma}$, its conventional construction requires exact knowledge of system matrices, as $A$ and $B$ explicitly appear in~\eqref{eq:SF2}. In this paper, however, these matrices are unknown, which constitutes the main challenge. The next section proposes a direct data-driven framework to address this problem.

\section{Data-Driven Methodology}\phantomsection\label{Sec:DD}

\subsection{Data Collection}\phantomsection\label{Subsec:data}
We first perform a finite-horizon experiment on the dt-LCS $\Sigma$ over the interval $[0,T]$, where $T \in \Np$ denotes the experiment horizon, and collect the input--state data
\begin{subequations}\phantomsection\label{eq:data}
	\begin{align}
		X & = \begin{bmatrix}
			x(0) & x(1) & \ldots & x(T - 1)
		\end{bmatrix} \in \R^{n \times T},\\
		U & = \begin{bmatrix}
			\nu(0) & \nu(1) & \ldots & \nu(T - 1)
		\end{bmatrix} \in \R^{m \times T},\\
		X_+ & = \begin{bmatrix}
			x(1) & x(2) & \ldots & x(T)
		\end{bmatrix} \in \R^{n \times T},\\
		W & = \begin{bmatrix}
			\varpi(0) & \varpi(1) & \ldots & \varpi(T - 1)
		\end{bmatrix} \in \R^{n \times T},
	\end{align}
\end{subequations}
where $W$ is unknown and cannot be measured directly. Since $W$ affects both $X$ and $X_+$ through the system dynamics, the collected data are noise-corrupted.
While $W$ is unknown, we impose the following assumption~\cite{9308978}, which indicates that the disturbance is bounded.

\begin{assumption}\phantomsection\label{assump:noise}
	There exists a known constant $\varepsilon \in \Rp$ such that, for all $w \in \mathds{W}$, $
	|w| \le \varepsilon.$
\end{assumption}

We remark that, in practical applications, a conservative yet reasonable value of $\varepsilon$ can be chosen based on physical insight into the source of the disturbance. Notice also that, under Assumption~\ref{assump:noise}, one has
\begin{align}
	W W^\top = \sum_{k = 0}^{T - 1} \varpi(k) \varpi^\top(k) \preceq \sum_{k = 0}^{T - 1} \vert\varpi(k)\vert^2 \I_n \preceq \Delta, \phantomsection\label{eq:WWT-bound}
\end{align}
where $\Delta \coloneq \varepsilon^2 T \I_n$.
We also impose the following assumption on the collected data in~\eqref{eq:data}, which facilitates the feasibility analysis of the proposed framework (cf., Section~\ref{Subsec:feasibility}).

\begin{assumption}\phantomsection\label{assump:rank}
	The collected data satisfy the rank condition
	\begin{align}
		\operatorname{rank} (H)=m+n, \phantomsection\label{eq:rank}
	\end{align}
	where $H = \begin{bmatrix}
		U^\top & X^\top
	\end{bmatrix}^{\! \top}$.
\end{assumption}

\begin{remark}\phantomsection\label{rem:assump2}
	Assumption~\ref{assump:rank}, which is closely related to the notion of persistency of excitation~\cite{willems2005note}, essentially serves as a data richness condition, ensuring that the collected data are sufficiently informative for the subsequent analysis. When the noise is sufficiently small, a sufficient condition for Assumption~\ref{assump:rank} to hold is that the input sequence is persistently exciting of order $n+1$ and that the pair $(A,B)$ is controllable~\cite{willems2005note}. Moreover, the rank condition~\eqref{eq:rank} implies that the experiment horizon should at least satisfy $T \geq m+n$.
\end{remark}

\subsection{Data-Driven ROM and SF Construction}\phantomsection\label{Subsec:result}
With the data collection procedure specified, we now turn to developing our direct data-driven framework. To construct an SF satisfying the conditions in Definition~\ref{def:SF}, we restrict attention to quadratic functions of the form
$
\mathcal{S}(x,\hat x)=(x-R\hat x)^\top P(x-R\hat x),
$
where \(P \succ 0\) and \(R \in \R^{n \times \hat n}\) denotes the \emph{reconstruction matrix}, which maps the ROM state $\hat x$ into the original state space, so that $R \hat x$ serves as the ROM-based approximation of $x$, and $x - R\hat x$ is the reconstruction error quantified by $\mathcal{S}$.

Given the quadratic structure adopted for the SF, the term $x^+ - R \hat x^+$ naturally arises in condition~\eqref{eq:SF2}. Motivated by this observation, the following lemma derives a suitable parameterization of this term and simultaneously presents the proposed structure of the interface function.

\begin{lemma}\phantomsection\label{lemma}
	Given $\Sigma = (\mathds{X}, \mathds{U}, \mathds{Y}, \mathds{W}, A, B, \I_n)$ and its ROM $\hat \Sigma = (\hat{ \mathds{X}}, \hat{ \mathds{U}}, \hat{ \mathds{Y}}, \hat A, \hat B, \hat C)$, let
	$
	S \coloneq \begin{bmatrix}
		B & A
	\end{bmatrix}
	$,
	and the interface function (a feedback law) be structured as 
	\begin{align}
		u = GP(x - R \hat x) + E \hat x + D \hat u, \phantomsection\label{eq:interface_map}
	\end{align}
	where \(G \in \R^{m \times n}\), \(P \in \R^{n \times n}\), \(R \in \R^{n \times \hat n}\), \(E \in \R^{m \times \hat n}\), and \(D \in \R^{m \times \hat m}\). Accordingly, one obtains the following parameterization:
	\begin{align}
		x^+ - R \hat x^+ & =  S \begin{bmatrix}
			GP\\
			\I_n
		\end{bmatrix} \! (x - R \hat x) + \Big(  S \begin{bmatrix}
		E\\
		R
		\end{bmatrix} - R \hat A  \Big) \hat x \notag\\
		& ~~~ + \Big(  S \begin{bmatrix}
		D\\
		\bbzero_{n \times \hat m}
		\end{bmatrix} - R \hat B  \Big) \hat u  + w. \phantomsection\label{eq:param}
	\end{align}
\end{lemma}

\begin{proof}
	Given that $x^+ \coloneq Ax + Bu + w$, $\hat x^+ \coloneq \hat A \hat x + \hat B \hat u$, and $
	S \coloneq \begin{bmatrix}
		B & A
	\end{bmatrix}
	$, and considering the interface function in~\eqref{eq:interface_map}, one has
	\begin{align*}
		x^+ - R \hat x^+ & = Ax + Bu + w - R (\hat A \hat x + \hat B \hat u)\\
		& \! \overset{\eqref{eq:interface_map}}{=} Ax + BGP(x - R \hat x) + BE\hat x + BD\hat u + w\\ & ~~~ - R (\hat A \hat x + \hat B \hat u)\\
		& = A (x - R \hat x) + BGP(x - R \hat x) + AR\hat x + BE\hat x \\ & ~~~ + BD\hat u + w - R\hat A \hat x - R \hat B \hat u\\
		& = (A + BGP) (x - R \hat x) + (AR+BE-R\hat A)\hat x\\
		& ~~~ + (BD - R\hat B) \hat u + w\\
		& = \begin{bmatrix}
			B & \!\!\! A
		\end{bmatrix} \!\! \begin{bmatrix}
		GP \\
		\I_n
		\end{bmatrix} \!\! (x  -  R \hat x) + \Big( \! \begin{bmatrix}
		B & \!\!\! A
		\end{bmatrix} \!\! \begin{bmatrix}
		E \\
		R
		\end{bmatrix}\!\! -R\hat A  \Big)\hat x\\
		& ~~~ +  \Big( \! \begin{bmatrix}
			B & \!\!\! A
		\end{bmatrix} \!\! \begin{bmatrix}
			D \\
			\bbzero_{n \times \hat m}
		\end{bmatrix}\!\! -R\hat B  \Big)\hat u + w,
	\end{align*}
	where the third equality is obtained by addition and subtraction of the term $AR\hat x$, leading to~\eqref{eq:param} with $
	S \coloneq \begin{bmatrix}
		B & A
	\end{bmatrix}
	$, thereby concluding the proof.
\end{proof}

With the interface function in~\eqref{eq:interface_map} and the parameterization in~\eqref{eq:param} at hand, we now present the following theorem, which constitutes the main result of the paper and allows the interface function~\eqref{eq:interface_map} and the SF $\mathcal{S}(x,\hat x)$ to be constructed simultaneously from noise-corrupted data without identifying the system matrices.

\begin{theorem}\phantomsection\label{thm:main}
	Given $\Sigma = (\mathds{X}, \mathds{U}, \mathds{Y}, \mathds{W}, A, B, \I_n)$ and its ROM $\hat \Sigma = (\hat{ \mathds{X}}, \hat{ \mathds{U}}, \hat{ \mathds{Y}}, \hat A, \hat B, \hat C)$ with $\hat C = R$, where $R$ is to be constructed, let Assumptions~\ref{assump:noise} and~\ref{assump:rank} hold. Consider the following semidefinite program (SDP) in the decision variables $\beta \in \Rp$, $K_1 \in \R^{T \times \hat n}$, $K_2 \in \R^{T \times \hat m}$, $\bar{P} \in \R^{n \times n}$, with $\bar{P} \succ 0$, $G \in \R^{m \times n}$, and $\bar\mu \in \Rpz$, for some fixed $\eta, \, \mu_i \in \Rp$, $i \in \{1,\ldots, 6\}$, $0 < \kappa < 1$, $\hat A$, and $\hat B$:
	\begin{mini!}|s|[2]<b>
		{\substack{\bar{\mu}, K_1, K_2,\\ G, \bar{P}, \beta}}{\Vert K_1 \Vert  \! + \! \Vert K_2 \Vert \! + \! \Vert X_+ K_2 \! - \! X K_1 \hat B \Vert \! + \! \beta}
		{\label{eq:SDP}}{\label{eq:mini_SDP}}
		\addConstraint{	\One_T^{\! \top} \, K_1 \, \One_{\hat n} \geq \eta,}{ \label{eq:SDP_con1} }
		\addConstraint{\bar{P} \preceq \beta \I_n,}{ \label{eq:SDP_con2} }
		\addConstraint{X_+ K_1 = X K_1 \hat A,}{ \label{eq:SDP_con3} }
		\addConstraint{X K_2 = \bbzero_{n \times \hat m},}{ \label{eq:SDP_con5} }
		\addConstraint{
		\mathcal{Q}_1 - \bar{\mu} \mathcal{Q}_2 \succeq 0,
		}{ \label{eq:SDP_con6} }
	\end{mini!}
	where
	\begin{subequations}
		\begin{align}
			\mathcal{Q}_1 & \coloneq \begin{bmatrix}
				\kappa \bar{P} & \bbzero_{n \times (n + m)} & \bbzero_{n \times n}\\
				\star & \bbzero_{(n + m) \times (n + m)} & \begin{bmatrix}
					G\\
					\bar{P}
				\end{bmatrix}\\
				\star & \star & \frac{1}{1 + \sum_{i = 1}^{3} \mu_i} \bar{P}
			\end{bmatrix}\!\!,\phantomsection\label{eq:Q1}\\
			\mathcal{Q}_2 & \coloneq \begin{bmatrix}
				\Delta - X_+X_+^\top & X_+H^\top & \bbzero_{n \times n}\\
				\star & -HH^\top & \bbzero_{(n+m)\times n}\\
				\star & \star & \bbzero_{n \times n}
			\end{bmatrix}\!\!,\phantomsection\label{eq:Q2}
		\end{align}
	\end{subequations}
	with $H = \begin{bmatrix}
		U^\top & \!\!\! X^\top
	\end{bmatrix}^{\! \top}$. If the SDP~\eqref{eq:SDP} is feasible, then any feasible solution yields the interface function~\eqref{eq:interface_map}, with $P \coloneq \bar{P}^{-1}$, $R \coloneq XK_1$, $E \coloneq UK_1$, and $D \coloneq UK_2$. Moreover, $\mathcal{S}(x,\hat x)=(x-R\hat x)^\top P(x-R\hat x)$ is the corresponding SF from $\hat{\Sigma}$ to $\Sigma$, with $\alpha \coloneq \lambda_{\min}(P)$, and
	\begin{subequations}\phantomsection \label{eq:rho_psi}
		\begin{align}
			\rho &  \coloneq \! \big(1 + \frac{1}{\mu_2} + \frac{1}{\mu_4} + \mu_6\big) \lambda_{\max}(P) \mathcal{Z} ,\phantomsection\label{eq:rho}\\
			\psi &  \coloneq \! \big(1 + \frac{1}{\mu_1} + \mu_4 + \mu_5\big) \lambda_{\max}(P) \lambda_{\max}(\Delta) \Vert K_1 \Vert^2 \max_{\hat x \in \hat{ \mathds{X}}} \vert \hat x \vert^2\notag\\
			& ~~~ +  \big(1 + \frac{1}{\mu_3} + \frac{1}{\mu_5} + \frac{1}{\mu_6}\big) \lambda_{\max}(P) \varepsilon^2, \phantomsection\label{eq:psi}
		\end{align}
	\end{subequations}
	where $\mathcal{Z} \coloneq (\Vert X_+ K_2 - X K_1 \hat B \Vert + \sqrt{\lambda_{\max}(\Delta)} \Vert K_2 \Vert)^2$.
\end{theorem}

\begin{proof}
	The proof is carried out in two steps. For any feasible solution of the SDP~\eqref{eq:SDP}, the first step illustrates that condition~\eqref{eq:SF1} is satisfied, whereas the second, and technically more involved, step establishes condition~\eqref{eq:SF2}. To continue with the first step, recall that $\hat C = R$, which yields
	$
	\vert x - \hat C \hat x\vert^2 = \vert x - R \hat x\vert^2.
	$
	Moreover, since
	\[
	\lambda_{\min}(P) \vert x - R \hat x\vert^2 \leq (x- R\hat x)^\top P (x- R\hat x) = \mathcal{S}(x,\hat{x}),
	\]
	it follows that~\eqref{eq:SF1} holds with $\alpha \coloneq \lambda_{\min}(P)$. This completes the first stage of the proof.
	
	To continue with the second step, considering $\mathcal{S}(x,\hat x)=(x-R\hat x)^\top P(x-R\hat x)$, one has
	\begin{align}
		\mathcal{S}(x^+, \hat x^+) & = (x^+ -R \hat x^+)^\top P (x^+ -R \hat x^+) \notag\\
		& \! \overset{\eqref{eq:param}}{=}  \Big(  S  \begin{bmatrix}
			G\\
			P^{-1}
		\end{bmatrix}  P  (x - R \hat x)  + \overbrace{\big (S \begin{bmatrix}
			E\\
			R
		\end{bmatrix} - R \hat A  \big)}^{\clubsuit } \hat x	 \notag\\
		& ~~~~~ + \overbrace{\big(  S \begin{bmatrix}
				D\\
				\bbzero_{n \times \hat m}
			\end{bmatrix} - R \hat B  \big)}^{\spadesuit } \hat u  + w\Big)^{\! \! \top} P \, (\ast),\phantomsection\label{tmp1}
	\end{align}
	where $(\ast)$ represents the expression inside the large parentheses in~\eqref{tmp1}.
	Observe that, according to the dt-LCS $\Sigma$ in~\eqref{eq:dt-LCS}, the collected data in~\eqref{eq:data} satisfy
	\begin{align}
		X_+ = A X + B U + W = S H + W, \phantomsection\label{tmp2}
	\end{align}
	where $
	S \coloneq \begin{bmatrix}
		B & \!\!\! A
	\end{bmatrix}
	$ and $H \coloneq \begin{bmatrix}
		U^\top & \!\!\! X^\top
	\end{bmatrix}^{\! \top}$.
	We first restrict our attention to the term ``$\clubsuit$'' in~\eqref{tmp1}. Inspired by the insightful work~\cite{mao2025oneTAC}, and recalling that $R \coloneq XK_1$ and $E \coloneq UK_1$, under constraint~\eqref{eq:SDP_con3}, one has
	\begin{align*}
		X_+ K_1 = X K_1 \hat A & \overset{\eqref{tmp2}}{\iff} SHK_1 + WK_1 = XK_1 \hat A\\
		& \iff S \begin{bmatrix}
			UK_1\\XK_1
		\end{bmatrix} + WK_1 = XK_1 \hat A\\
		& \iff S \begin{bmatrix}
			E\\R
		\end{bmatrix} + WK_1 = R \hat A.
	\end{align*}
	Hence, the term ``$\clubsuit$'' in~\eqref{tmp1} can be substituted by $-WK_1$.
	
	Focusing on the term ``$\spadesuit$'' in~\eqref{tmp1} and recalling that $D \coloneq U K_2$ and $R \coloneq XK_1$, under constraint~\eqref{eq:SDP_con5}, we obtain
	\begin{align}
		S \begin{bmatrix}
			D\\
			\bbzero_{n \times \hat m}
		\end{bmatrix} - R \hat B &  \overset{\eqref{eq:SDP_con5}}{=} S \overbrace{\begin{bmatrix}
			U\\
			X
		\end{bmatrix}}^H K_2 - R \hat B\notag\\
		& \, \overset{\eqref{tmp2}}{=} (X_+ - W)K_2 - XK_1 \hat B\notag\\
		& ~ = \overbrace{(X_+K_2 - XK_1\hat B)}^{\mathcal{N}_1} \overbrace{-WK_2}^{\mathcal{N}_2} .\phantomsection\label{eq:tmpn1}
	\end{align}
	Therefore, the term ``$\spadesuit$'' in~\eqref{tmp1} can be replaced by~\eqref{eq:tmpn1}.
	Incorporating these substitutions, one can rewrite~\eqref{tmp1} as
	\[
	\begin{aligned}
		& \mathcal{S}(x^+, \hat x^+)\\
		&=    \! \Big(   \overbrace{\!\!  S  \!   \begin{bmatrix}
				G\\
				P^{-1}
			\end{bmatrix}  \!\!  P  (x \!  - \!   R \hat x)}^{\mathcal{T}_1}   \overbrace{-   WK_1 \hat x}^{\mathcal{T}_2}  +\! \overbrace{(\mathcal{N}_1 \! + \! \mathcal{N}_2) \hat u}^{\mathcal{T}_3} \!  +\!     \overbrace{w}^{\mathcal{T}_4} \Big)^{ \!\!  \top} \! \! P (\ast)\\
		& = \sum_{i=1}^{4} \mathcal{T}_i^\top P \mathcal{T}_i + \sum_{j = 2}^{4} 2 \mathcal{T}_1^\top P \mathcal{T}_j + \sum_{z = 3}^{4} 2 \mathcal{T}_2^\top P \mathcal{T}_z + 2 \mathcal{T}_3^\top P \mathcal{T}_4.
	\end{aligned}
	\]
	Notice that for all $i,j\in\{1, \ldots, 4\}$, one can rewrite $2\mathcal{T}_i^\top P \mathcal{T}_j$ as $2 (\sqrt{P}\mathcal{T}_i)^\top (\sqrt{P} \mathcal{T}_j)$.
	Considering this, we next employ the Cauchy–Schwarz inequality, \emph{i.e.}, $a^\top b \le |a| |b|$ for any $a,b \in \R^{n}$, and subsequently apply Young’s inequality, \emph{i.e.}, $|a| |b| \le \tfrac{\mu_i}{2}|a|^2 + \tfrac{1}{2\mu_i}|b|^2$ for any $\mu_i \in \Rp$ for all $i \in \{1, \ldots, 6\}$, which yields
	\[
	\begin{aligned}
		& \sum_{j = 2}^{4} 2 \mathcal{T}_1^\top P \mathcal{T}_j  \leq \sum_{i=1}^{3}\mu_i \mathcal{T}_1^\top P \mathcal{T}_1 + \frac{1}{\mu_1} \mathcal{T}_2^\top P \mathcal{T}_2 + \frac{1}{\mu_2} \mathcal{T}_3^\top P \mathcal{T}_3 \\
		& \hspace{2.2cm} + \frac{1}{\mu_3} \mathcal{T}_4^\top P \mathcal{T}_4, \\
		&  \sum_{z = 3}^{4}  2 \mathcal{T}_2^\top P \mathcal{T}_z  \leq  \sum_{i=4}^{5}  \mu_i \mathcal{T}_2^\top P \mathcal{T}_2 + \frac{1}{\mu_4} \mathcal{T}_3^\top P \mathcal{T}_3 + \frac{1}{\mu_5} \mathcal{T}_4^\top P \mathcal{T}_4,\\
		& 2 \mathcal{T}_3^\top P \mathcal{T}_4 \leq \mu_6 \mathcal{T}_3^\top P \mathcal{T}_3  + \frac{1}{\mu_6} \mathcal{T}_4^\top P \mathcal{T}_4.
	\end{aligned}
	\]
	Consequently, we have
		\[
	\begin{aligned}
		& \mathcal{S}(x^+, \hat x^+)\\
		&\leq \big(1 + \sum_{i=1}^{3}\mu_i \big) \mathcal{T}_1^\top P \mathcal{T}_1 + \big( 1 + \frac{1}{\mu_1} + \mu_4 + \mu_5 \big) \mathcal{T}_2^\top P \mathcal{T}_2\\
		& ~ + \!\! \big( 1 \! + \! \frac{1}{\mu_2} \! + \! \frac{1}{\mu_4} \! +\! \mu_6 \big) \mathcal{T}_3^\top P \mathcal{T}_3 \! + \! \big( 1 \! +\! \frac{1}{\mu_3} \! +\! \frac{1}{\mu_5} \!+\! \frac{1}{\mu_6} \big) \mathcal{T}_4^\top P \mathcal{T}_4.
	\end{aligned}
	\]
	We now proceed to derive upper bounds for $\mathcal{T}_2^\top P \mathcal{T}_2$, $\mathcal{T}_3^\top P \mathcal{T}_3$, and $\mathcal{T}_4^\top P \mathcal{T}_4$, which eventually lead to $\rho$ and $\psi$ as specified in~\eqref{eq:rho_psi}.
	To this end, we have
	\[
	\begin{aligned}
		\mathcal{T}_2^\top P \mathcal{T}_2 & = \hat x^\top K_1^\top W^\top P W K_1 \hat x\\
		& \leq \lambda_{\max}(P) \hat x^\top K_1^\top W^\top W K_1 \hat x \\
		& \leq  \lambda_{\max}(P) \lambda_{\max}(W^\top W) \hat x^\top K_1^\top K_1 \hat x\\
		& = \lambda_{\max}(P) \lambda_{\max}(WW^\top) \hat x^\top K_1^\top K_1 \hat x\\
		& \! \overset{\eqref{eq:WWT-bound}}{\leq} \lambda_{\max}(P) \lambda_{\max}(\Delta) \Vert K_1 \Vert^2 \vert \hat x \vert^2\\
		& \leq \lambda_{\max}(P) \lambda_{\max}(\Delta) \Vert K_1 \Vert^2 \, \max_{\hat x \in \hat{ \mathds{X}}} \vert \hat x \vert^2.
	\end{aligned}
	\]
	Following the same rationale and recalling Assumption~\ref{assump:noise}, one has $\mathcal{T}_4^\top P \mathcal{T}_4 \leq  \lambda_{\max}(P) \varepsilon^2$. As for $\mathcal{T}_3^\top P \mathcal{T}_3$, we have
	\[
	\begin{aligned}
		\mathcal{T}_3^\top P \mathcal{T}_3 & = \hat{u}^\top (\mathcal{N}_1 + \mathcal{N}_2)^\top P (\mathcal{N}_1 + \mathcal{N}_2) \hat{u}\\
		& \leq \lambda_{\max}(P) \Vert \mathcal{N}_1 + \mathcal{N}_2 \Vert^2 \vert \hat{u} \vert^2 \\
		& \leq \lambda_{\max}(P) (\Vert \mathcal{N}_1 \Vert + \Vert \mathcal{N}_2 \Vert)^2 \vert \hat{u} \vert^2\\
		& \! \overset{\eqref{eq:tmpn1}}{\leq}  \lambda_{\max}(P) (\Vert \mathcal{N}_1 \Vert + \Vert W \Vert \Vert K_2 \Vert)^2 \vert \hat{u} \vert^2\\
		& \! \overset{\eqref{eq:WWT-bound}}{\leq}  \lambda_{\max}(P) \overbrace{(\Vert \mathcal{N}_1 \Vert + \sqrt{\lambda_{\max}(\Delta)} \Vert K_2 \Vert)^2}^{\mathcal{Z}} \vert \hat{u} \vert^2,
	\end{aligned}
	\]
	where, for the last inequality, we use
	\[
	\Vert W \Vert^2 = \lambda_{\max}(WW^\top) \overset{\eqref{eq:WWT-bound}}{\leq}  \lambda_{\max}(\Delta) \! \Rightarrow \! \Vert W \Vert \leq \sqrt{\lambda_{\max}(\Delta)}. 
	\]
	Accordingly, one obtains
	\[
	\begin{aligned}
		& \mathcal{S}(x^+, \hat x^+)\\
		&\leq \big(1 + \sum_{i=1}^{3}\mu_i \big) \mathcal{T}_1^\top P \mathcal{T}_1\\ & ~~~+ \big( 1 + \frac{1}{\mu_1} + \mu_4 + \mu_5 \big) \lambda_{\max}(P) \lambda_{\max}(\Delta) \Vert K_1 \Vert^2 \, \max_{\hat x \in \hat{ \mathds{X}}} \vert \hat x \vert^2\\
		& ~~~ + \big( 1   +   \frac{1}{\mu_2}   +   \frac{1}{\mu_4}   +  \mu_6 \big) \lambda_{\max}(P) \mathcal{Z} \vert \hat{u} \vert^2 \\  &~~~+   \big( 1   +  \frac{1}{\mu_3}   +  \frac{1}{\mu_5}  +  \frac{1}{\mu_6} \big) \lambda_{\max}(P) \varepsilon^2,
	\end{aligned}
	\]
	establishing $\rho$ and $\psi$ as in~\eqref{eq:rho_psi} (cf., condition~\eqref{eq:SF2}).
	
	To proceed with the proof, it remains to show that under constraint~\eqref{eq:SDP_con6}, the inequality
	\begin{align}
		\big(1 + \sum_{i=1}^{3}\mu_i \big) \mathcal{T}_1^\top P \mathcal{T}_1 \leq \kappa(x - R\hat x)^\top P (x - R \hat x) \phantomsection\label{TMP}
	\end{align}
	holds. To show this, recalling $P = P^\top \succ 0$, it suffices to establish that
	\[
	\kappa P - \big(1 + \sum_{i=1}^{3}\mu_i \big) P \begin{bmatrix}
		G\\
		P^{-1}
	\end{bmatrix}^{\!\!\top} S^\top P  S \begin{bmatrix}
	G\\
	P^{-1}
	\end{bmatrix} P \succeq 0.
	\]
	Pre-multiplying and post-multiplying this by $P^{-1}$, one gets
	\[
	\kappa P^{-1} - \big(1 + \sum_{i=1}^{3}\mu_i \big) \begin{bmatrix}
		G\\
		P^{-1}
	\end{bmatrix}^{\!\!\top} S^\top P  S \begin{bmatrix}
		G\\
		P^{-1}
	\end{bmatrix} \succeq 0.
	\]
	By the Schur complement, this expression is equivalent to
	\[
	\begin{aligned}
		& \begin{bmatrix}
			\frac{1}{1 + \sum_{i = 1}^{3} \mu_i} P^{-1}& S \begin{bmatrix}
				G\\
				P^{-1}
			\end{bmatrix}\\
			\star& \kappa P^{-1}
		\end{bmatrix} \succeq 0\\
		& \iff \begin{bmatrix}
			\kappa P^{-1}& \begin{bmatrix}
				G\\
				P^{-1}
			\end{bmatrix}^{\!\!\top} S^\top\\
			\star& \frac{1}{1 + \sum_{i = 1}^{3} \mu_i} P^{-1}
		\end{bmatrix} \succeq 0\\
		& \iff \kappa P^{-1} - \big(1 + \sum_{i=1}^{3}\mu_i \big) S \begin{bmatrix}
			G\\
			P^{-1}
		\end{bmatrix} P \begin{bmatrix}
			G\\
			P^{-1}
		\end{bmatrix}^{\!\!\top} S^\top  \succeq 0,
	\end{aligned}
	\]
	which can subsequently be written as
	\begin{align}
		\begin{bmatrix}
			\I_n\\S^\top
		\end{bmatrix}^{\! \! \top} \!\! \overbrace{\begin{bmatrix}
			\kappa P^{-1} & \bbzero_{n \times (n + m)}\\
			\star & -\big(1 \! + \! \sum_{i=1}^{3}\mu_i \big) \!\! \begin{bmatrix}
				G\\
				P^{-1}
			\end{bmatrix} \!\! P \! \begin{bmatrix}
				G\\
				P^{-1}
			\end{bmatrix}^{\!\!\top}
			\end{bmatrix}}^{\mathcal{M}_1} \!\! \! \begin{bmatrix}
		\I_n\\S^\top
		\end{bmatrix}\!\! \succeq 0. \phantomsection\label{tmp3}
	\end{align}
	Satisfying~\eqref{tmp3} presents two challenges: \emph{(i)} it depends on the unknown system matrices (since $S \coloneq \begin{bmatrix}
			B & \!\!\! A
		\end{bmatrix}$), and \emph{(ii)} its second diagonal block is nonlinear, which obstructs solving it.
	To overcome the first difficulty, we return to Assumption~\ref{assump:noise} together with the linear matrix inequality (LMI)~\eqref{eq:WWT-bound}. Specifically, the LMI~\eqref{eq:WWT-bound} can be rewritten as
	\[
	\begin{aligned}
		\Delta \succeq WW^\top & \overset{\eqref{tmp2}}{\iff}  \Delta \succeq (X_+ - SH) (X_+ - SH)^\top,
	\end{aligned}
	\]
	which yields
	\begin{align}
		\begin{bmatrix}
			\I_n\\S^\top
		\end{bmatrix}^{\! \! \top}  \overbrace{\begin{bmatrix}
			\Delta - X_+ X_+^\top &  X_+ H^\top\\
			\star & - H H^\top
			\end{bmatrix}}^{\mathcal{M}_2}  \begin{bmatrix}
			\I_n\\S^\top
		\end{bmatrix}\!\! \succeq 0. \phantomsection\label{tmp4}
	\end{align}
	Since both LMIs~\eqref{tmp3} and~\eqref{tmp4} are quadratic in $\begin{bmatrix} \I_n & \!\!\! S \end{bmatrix}^{\!\!\top}$, the classical S-procedure~\cite[Theorem~9]{9308978} can be employed to enforce~\eqref{tmp3} while accounting for~\eqref{tmp4}. In particular, this removes the need for explicit knowledge of $S$, and~\eqref{tmp3} can be guaranteed provided that there exists a multiplier $\bar{\mu} \in \Rpz$ such that
	\begin{align}
		\mathcal{M}_1 - \bar{\mu} \mathcal{M}_2 \succeq 0.\phantomsection\label{tmp5}
	\end{align}
	With the first challenge resolved, it remains to address the second one. For this purpose, recalling $\bar{P}=P^{-1}$, we can again apply the Schur complement to show that
	\[
	\eqref{eq:SDP_con6} \iff \eqref{tmp5}.
	\]
	Therefore, under constraint~\eqref{eq:SDP_con6}, condition~\eqref{TMP} is ensured. Combining~\eqref{TMP} with the preceding bound yields 
	\[
	\mathcal{S}(x^+,\hat x^+) \leq \kappa\mathcal{S}(x,\hat x)+\rho|\hat u|^2+\psi,
	\]
	which is condition~\eqref{eq:SF2}, with $\rho$ and $\psi$ as in~\eqref{eq:rho_psi}, thereby concluding the proof.
\end{proof}

\begin{remark}\phantomsection\label{rem:on_n}
	Smaller values of $\hat{n}$ typically facilitate controller synthesis for $\hat{\Sigma}$. The choice of $\hat{n}$ mainly depends on two aspects: \emph{(i)} the specification imposed on $\Sigma$ and \emph{(ii)} the feasibility of achieving a meaningful closeness guarantee. For instance, if the objective concerns only the first two state variables of the original system (\emph{e.g.}, positions along the $x$- and $y$-axes), selecting $\hat{n}=2$ is desirable, since it enables direct control. If such a reduction does not yield a favorable closeness guarantee, larger values of $\hat{n}$ should be considered progressively until a suitable $\hat{\Sigma}$ is obtained.
\end{remark}

\begin{remark}\phantomsection \label{Remark_ hatB}
	If $R$ and $\hat{A}$ were to be designed simultaneously, constraint~\eqref{eq:SDP_con3} would become bilinear due to their product. Such an issue is inherent to SF-based methods and also appears in model-based settings (cf., condition~(20a) in~\cite{zamani2017compositional}). To circumvent this issue, for a chosen $\hat{n}$ (cf., Remark~\ref{rem:on_n}), one can fix $\hat{A}$ and then solve the SDP~\eqref{eq:SDP}. In practice, selecting $\hat{A}$ is straightforward; choosing it as a Schur matrix simplifies the control design for $\hat{\Sigma}$. Moreover, since Theorem~\ref{thm:main} imposes no restriction on the structure of $\hat B$, a natural choice is $\hat{B} = \gamma \I_{\hat n}$ with $\gamma \in \R \backslash \{0\}$ (hence $\hat m = \hat n$). This renders $\hat{\Sigma}$ fully actuated, which facilitates controller synthesis. Additionally, through its effect on the term $\Vert X_+ K_2 - X K_1 \hat B \Vert$ in~\eqref{eq:rho}, the parameter $\gamma$ can be used to regulate $\rho$, and thus the closeness guarantee in~\eqref{eq:closeness}.
\end{remark}

\section{Discussion}\phantomsection\label{Sec:Disc}

\subsection{Feasibility Analysis of SDP~\eqref{eq:SDP}}\phantomsection\label{Subsec:feasibility}
We first note that constraints~\eqref{eq:SDP_con1} and~\eqref{eq:SDP_con2} are not intrinsic to the proposed framework. In particular, since $K_1 = \bbzero_{T \times \hat n}$ is a trivial yet undesirable solution of~\eqref{eq:SDP_con3}, constraint~\eqref{eq:SDP_con1} is introduced solely to exclude this case. This constraint is not unique and can be replaced by other constraints (\emph{e.g.}, a favorable constraint on $XK_1$).
Likewise, constraint~\eqref{eq:SDP_con2} is introduced, and $\beta \in \Rp$ is included in the cost function~\eqref{eq:mini_SDP} to promote a potentially large value of $\alpha \coloneq \lambda_{\min}(P)$, which may tighten the closeness guarantee~\eqref{eq:closeness}.
Similarly, the inclusion of $\Vert K_1 \Vert$, $\Vert K_2 \Vert$, and $\Vert X_+ K_2 - X K_1 \hat B \Vert$ in the cost function~\eqref{eq:mini_SDP} promotes solutions with smaller values of these norms, which may reduce the values of $\psi$ and $\rho$ in~\eqref{eq:rho_psi}, thereby further tightening the closeness guarantee~\eqref{eq:closeness}. Notice also that the strict constraint $\bar P \succ 0$ is enforced as $\bar P \succeq \delta \I_n$ for a fixed $\delta \in \Rp$.

We further note that constraint~\eqref{eq:SDP_con3} alone is a homogeneous system of linear equations and is therefore always feasible, since it admits the trivial solution $K_1=\bbzero_{T \times \hat n}$. Moreover, if $T>n$, which follows from Assumption~\ref{assump:rank} and Remark~\ref{rem:assump2} when $m \geq 1$, then~\eqref{eq:SDP_con3} has more unknowns than scalar equalities, namely, $T\hat n$ unknowns and $n\hat n$ scalar equalities; hence, it admits a nontrivial solution. Under the same condition $T>n$, constraint~\eqref{eq:SDP_con5} is also a homogeneous system of linear equations and therefore is always feasible while also admitting a nontrivial solution.

We also remark that Assumption~\ref{assump:rank} facilitates the feasibility of condition~\eqref{eq:SDP_con6}. In fact, Assumption~\ref{assump:rank} implies that $H$ has full row rank, and hence $HH^\top \succ 0$. Consequently, whenever $\bar{\mu} \in \Rp$, $\bar{\mu}HH^\top \succ 0$. If Assumption~\ref{assump:rank} fails, then $HH^\top$ becomes singular, and feasibility of~\eqref{eq:SDP_con6}, although still possible, requires the off-diagonal blocks coupled with $\bar{\mu}HH^\top$ to satisfy additional compatibility conditions; in particular, their columns should belong to $\operatorname{col}(HH^\top)$. Therefore, Assumption~\ref{assump:rank} removes this structural difficulty whenever $\bar{\mu}\in \Rp$, thereby facilitating the satisfaction of~\eqref{eq:SDP_con6}. Remark that the above discussions concern individual constraints and do not imply their joint feasibility. We also note that using the S-procedure to certify the constructed SF and interface function for all system matrices consistent with the collected data and the assumed disturbance bound may introduce conservatism, potentially affecting both the feasibility of the SDP and the tightness of the resulting closeness bound.

\subsection{Scalability Analysis}
We focus here on constraint~\eqref{eq:SDP_con6}, which is a central LMI in the proposed formulation; this constraint has dimension $(3n+m)\times(3n+m)$. Assuming that $T$, $m$, $\hat n$, and $\hat m$ grow at most linearly with $n$, the remaining constraints and norm terms do not change the following complexity orders. According to~\cite{8619019}, the resulting SDP has approximate time and memory complexities of $\mathcal{O}(n^{6.5})$ and $\mathcal{O}(n^{4})$, respectively, which can be computationally demanding. We note that, while energy-based and Krylov-based methods generally offer better scalability, the present framework provides an explicit interface-based controller-refinement mechanism and a trajectory-wise closeness certificate suited to complex LTL specifications (cf. the case study), features not generally provided by those methods.

\subsection{Limitations}
Similar to any methodology, the proposed framework has certain limitations, which can also be viewed as directions for future research. First, the current setting is specifically tailored to linear dynamical systems. We note that, for a particular class of nonlinear dynamical systems, \cite{samari2025dataARXIV} proposes a data-driven approach for constructing SFs and interface functions. However, the approach in~\cite{samari2025dataARXIV} is developed for the continuous-time setting and, more importantly, does not directly account for process disturbances. This suggests that extending the present framework to handle certain classes of nonlinear systems subject to process disturbances within the discrete-time setting would require further theoretical development. Second, the proposed data-driven framework requires input--state data and assumes that all state variables are directly measurable. Inspired by the recent work in~\cite{li2026controller}, one potential direction is to extend the current framework to settings in which the state vector is partially measurable. This direction is challenging in our MOR setting and requires technical development, as the current SDP relies on the full-state data matrices $X$ and $X_+$ to form the reconstruction matrix $R$ (cf.~\eqref{eq:SDP_con3}), the SF $\mathcal{S}$, and the components of the interface function in~\eqref{eq:interface_map}. Moreover, the interface function itself explicitly depends on the state $x$, which, along with $X$ and $X_+$, is not directly available under partial state measurements.

\section{Simulation Results}\phantomsection\label{Sec:simul}
This section demonstrates the effectiveness of the proposed data-driven framework. In particular, for a six-dimensional dt-LCS subject to process disturbances, whose first two state variables are assumed to represent its position along the $x$- and $y$-axes, we aim to construct a two-dimensional ROM, design a controller for the ROM using SCOTS~\cite{rungger2016scots} to satisfy a challenging reach-while-avoid specification, and subsequently refine the synthesized controller back to the original system, thereby enabling the dt-LCS to satisfy the desired specification.
It is worth noting that for the original dt-LCS, the direct design of such a controller using tools commonly employed in the community of formal methods, such as SCOTS, is remarkably challenging, if not infeasible. In contrast, employing the proposed MOR framework significantly facilitates this design task. Note that although the framework is capable of handling substantially higher-dimensional systems (see, \emph{e.g.}, the first case study in~\cite{samari2025dataARXIV}), we consider a six-dimensional system here for ease of presentation.
We performed all simulations using \textsc{Matlab} \textit{R2023b} on a MacBook Pro (Apple M2 Max, 32\,GB memory).

We consider a dt-LCS as in~\eqref{eq:dt-LCS}, where
\[
A = \begin{bmatrix}
	0.82 & 0.10 & 0 & 0 & 0 & 0\\
	0 & 0.78 & 0.12 & 0 & 0 & 0\\
	0 & 0 & 0.75 & 0.10 & 0 & 0\\
	0 & 0 & 0 & 0.72 & 0.08 & 0\\
	0 & 0 & 0 & 0 & 0.70 & 0.10\\
	0.05 & 0 & 0 & 0 & 0 & 0.68
\end{bmatrix} \!\!,
\]
and
\[
B = \begin{bmatrix}
	0.68 & 0.34 & 0.17 & 0 & 0 & 0\\
	0 & 0 & 0 & 0.34 & 0.68 & 0.34
\end{bmatrix}^{\!\top}\!\!\!\!,
\]
with both matrices being unknown. Moreover, the disturbance is generated as $\varpi(k)=\xi(k)\, [0~0~1~0~0~1]^\top$, where $\xi(k)$ is sampled independently and uniformly from $[-10^{-3},10^{-3}]$ at each time step. Consequently, Assumption~\ref{assump:noise} holds with the known but conservative bound $\varepsilon= 1.5 \times 10^{-3}$. Accordingly, we conduct an experiment on the dt-LCS for data collection with a horizon of $T = 20$, during which the components of $\nu(k)$ are sampled independently at each time step from a normal distribution with mean $-15$ and standard deviation $30$. Consequently, one can compute $\Delta = 4.5 \times 10^{-5} \I_6$ according to~\eqref{eq:WWT-bound}.

To proceed with the design, following Remark~\ref{Remark_ hatB}, we fix
\[
\hat A =
	(1 - 10^{-6})  \,\I_2, \quad \hat B = 0.0001 \I_2.
\]
Moreover, we set $\kappa = 0.81$, $\mu_1 = 0.1540$, $\mu_2 = 0.0480$, $\mu_3 = 0.0479$, $\mu_4 = 0.3118$, $\mu_5 = 0.3111$, and $\mu_6 = 0.9977$.
In this case study, following the discussion in Section~\ref{Subsec:feasibility}, we replace constraint~\eqref{eq:SDP_con1} by fixing the leading $2 \times 2$ block of $X K_1$ to $0.5 \I_2$, while setting $\delta = 1$. Solving the resulting SDP takes approximately $0.35$ seconds and yields $\Vert K_1 \Vert = 0.0847$, $\Vert K_2 \Vert = 6.127 \times 10^{-6}$, $\Vert X_+ K_2 - X K_1 \hat B \Vert = 1.6701 \times 10^{-4}$, $\beta = 1.2072$, $\bar{\mu} = 22.7640$, and
\[
	\begin{aligned}
		GP =
		\left[
		\begin{array}{cccc}
			-0.9358 & -0.5370 & -0.2482 & -0.0174 \\
			-0.002 & -0.0164 & -0.0443 & -0.3713
		\end{array}
		\right.
		\\
		\left.
		\begin{array}{cc}
			0.0013 & -0.0111 \\
			-0.7364 & -0.4193
		\end{array}
		\right]\!\!.
	\end{aligned}
	\]
Furthermore, we subsequently compute
\[
	\begin{aligned}
		R & \! = \!\! \begin{bmatrix}
			0.5 & 0 & -0.375 & \! -1.1626 & -1.5361 & \! -0.5551\\ 0 & 0.5 & 1.125 & \! 2.9377 & 3.9219 & \!1.5899
		\end{bmatrix}^{\!\!\top}\\
		E & \! = \!\! \begin{bmatrix}
			0.1324 &  -0.0735\\ -0.5960  &  1.4964
		\end{bmatrix}\!\!, ~\! D \! = \! 10^{-3} \! \times \!  \begin{bmatrix}
			0.0296 &   0.0525\\ -0.2618  &  0.5921
		\end{bmatrix}\!\!,
	\end{aligned}
	\]
and form the interface function as in~\eqref{eq:interface_map}. We also obtain $\alpha = 0.8284$ and $\lambda_{\max}(P) = 1$. Accordingly, one can compute $\rho = 7.2623 \times 10^{-7}$. Now, considering $\hat{\mathds{X}} = [-6,6]^2$ and $\hat{\mathds{U}} = [-6,6]^2$, we compute
$\max_{\hat{x}\in\hat{\mathds{X}}} |\hat{x}|^2 = 72$ and $|\hat{\nu}|_\infty^2 \leq 72$,
yielding $\psi = 2.4735 \times 10^{-4}$ and an error bound of $0.0436$ from~\eqref{eq:closeness} with $\mathcal{S}(\mathrm{x},\hat{\mathrm{x}})=0$. The corresponding simulation results are depicted in Fig.~\ref{fig:example}. As illustrated, the complex \emph{reach-while-avoid} specification, which requires the system trajectories to start from the initial set \legendsquare{START}, reach the target set \legendsquare{TARGET}, and avoid collisions with the obstacles \legendsquare{OBSTACLES}, is satisfied across all $10$ simulation runs.

\begin{figure}[t!]
	\centering
	\begin{subfigure}[b]{0.49\columnwidth}
		\centering
		\includegraphics[width=\linewidth]{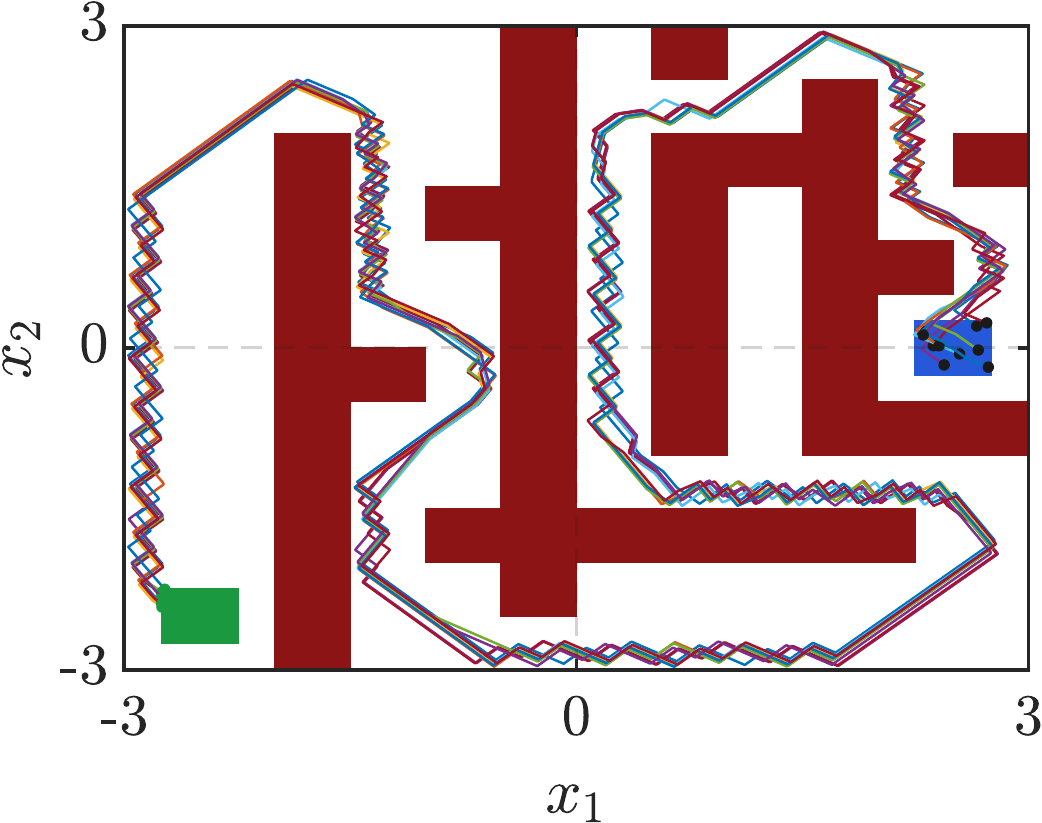}
		\caption{ }
		\label{fig:subfig2_1}
	\end{subfigure}
	\hfill
	\begin{subfigure}[b]{0.49\columnwidth}
		\centering
		\includegraphics[width=\linewidth]{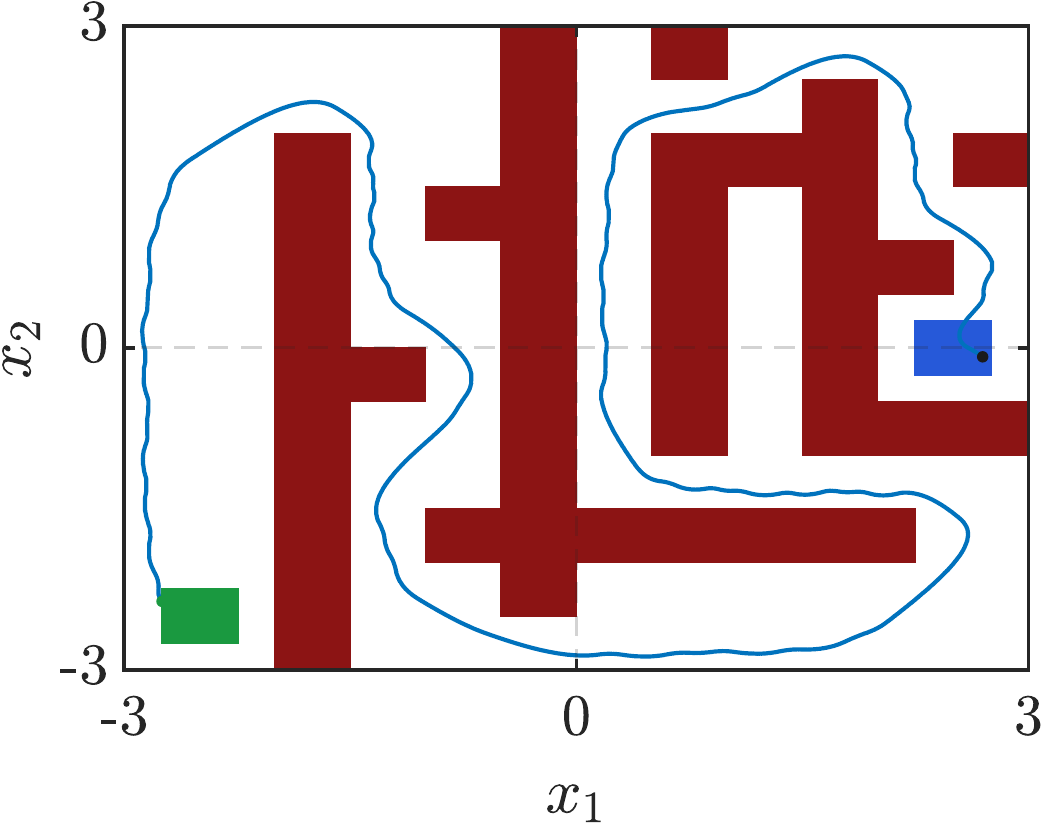}
		\caption{ }
		\label{fig:subfig2_2}
	\end{subfigure}
	
	\caption{(a) Trajectories of the dt-LCS for $10$ simulation runs with different initial conditions, all satisfying the complex specification.
		(b) A representative trajectory after applying a smoothing filter to mitigate the zigzag behavior observed in (a), which arises from the SCOTS synthesis tool due to the chosen discretization parameter.
		A video of this simulation is available at \url{https://youtu.be/iakgP6JJ3Bo}.
	}
	\label{fig:example}
\end{figure}

\section{Conclusion}
We proposed a direct data-driven framework for constructing ROMs of discrete-time linear dynamical systems with unknown dynamics and process disturbances. We derived data-dependent conditions for the simultaneous construction of ROMs, SFs, interface functions, and quantitative closeness guarantees directly from noise-corrupted input--state data collected along a single trajectory of the system. Future research will focus on extending the proposed framework to nonlinear systems and expanding it to incorporate input--output data.

\bibliographystyle{ieeetr}
\bibliography{biblio}

\end{document}